\documentclass[11pt]{article}
\usepackage[square, comma, sort&compress, numbers]{natbib}

\usepackage{multibib}

\usepackage{xcolor}

\usepackage{tikz}

\usepackage{enumerate}
\usepackage{bbm}
\usepackage{amsmath}
\usepackage{amssymb}

\usepackage{float}
\usepackage[export]{adjustbox}
\usepackage{graphicx}
\usepackage{subfigure}
\graphicspath{ {./images/} }
\usepackage[utf8]{inputenc}
\usepackage[export]{adjustbox}
\usepackage{wrapfig}
\usepackage{authblk}
\usepackage[top=27mm, bottom=27mm, left=22mm, right=22mm]{geometry}
\usepackage{amsthm,amsmath,amssymb}
\usepackage{mathrsfs}
\usepackage{multirow}
\usepackage{amsthm}
\usepackage{microtype}
\usepackage{hyperref}
\usepackage[capitalize]{cleveref}
\usepackage{xcolor}

\newtheorem{theorem}{Theorem}[section]
\newtheorem{lemma}[theorem]{Lemma}
\newtheorem{proposition}[theorem]{Proposition}

\newtheorem{definition}[theorem]{Definition}

\newtheorem{fact}[theorem]{Fact}

\DeclareMathOperator{\ca}{Cap}

\title{Constrained coding upper bounds via Goulden-Jackson cluster theorem\footnote{This project was supported by the National Key Research and Development Program of China under Grant 2020YFA0712100, the National Natural Science Foundation of China under Grant 12231014, Grant 12101364 and Grant 12271301, the Natural Science Foundation of Shandong Province under Grant ZR2021QA005, and Beijing Scholars Program.}}
\date{}

\makeatletter
\def\thanks#1{\protected@xdef\@thanks{\@thanks
        \protect\footnotetext{#1}}}
\makeatother

\author{Yuanting Shen, 
Chong Shangguan, 
Zhicong Lin, and 
Gennian Ge\thanks{Y. Shen, C. Shangguan, and Z. Lin are with Research Center for Mathematics and Interdisciplinary Sciences, Shandong University, Qingdao 266237, China, and Frontiers Science Center for Nonlinear Expectations, Ministry of Education, Qingdao 266237, China (Emails: shenyting121@163.com, theoreming@163.com, linz@sdu.edu.cn). G. Ge is with School of Mathematical Sciences, Capital Normal University, Beijing 100048, China (Email: gnge@zju.edu.cn).}
}

\begin{document}
\maketitle

\begin{abstract}
    \noindent Motivated by applications in DNA-based data storage, constrained codes have attracted a considerable amount of attention from both academia and industry. We study the maximum cardinality of constrained codes for which the constraints can be characterized by a set of forbidden substrings, where by a substring we mean some consecutive coordinates in a string.

    For finite-type constrained codes (for which the set of forbidden substrings is finite), one can compute their capacity (code rate) by the ``spectral method'', i.e., by applying the Perron-Frobenious theorem to the de Brujin graph defined by the code. However, there was no systematic method to compute the exact cardinality of these codes.

    We show that there is a surprisingly powerful method arising from enumerative combinatorics, which is based on the Goulden-Jackson cluster theorem (previously not known to the coding community), that can be used to compute not only the capacity, but also the exact formula for the cardinality of these codes, for each fixed code length. Moreover, this can be done by solving a system of linear equations of size equal to the number of constraints.

    We also show that the spectral method and the cluster method are inherently related by establishing a direct connection between the spectral radius of the de Brujin graph used in the first method and the convergence radius of the generating function used in the second method.

    Lastly, to demonstrate the flexibility of the new method, we use it to give an explicit upper bound on the maximum cardinality of variable-length non-overlapping codes, which are a class of constrained codes defined by an infinite number of forbidden substrings.
\end{abstract}



\section{Introduction}

\subsection{Finite-type constrained codes}\label{subsec:finite-type}

In storage and communication systems, data are encoded as strings. Since some substrings\footnote{Here by a substring we mean some consecutive coordinates in a string.} may have a higher error probability than others, it makes sense to avoid the appearance of them in the encoded strings. Given such a constraint, it is necessary to encode any data into strings that obey the constraint. Constrained codes have been widely used in storage and communication systems since the 1950s (see, e.g., \cite{immink1990runlength,immink1995efmplus,immink1998codes,immink2001survey,immink2004codes,marcus2001introduction}). Currently, it has been shown to have prominent applications in DNA-based data storage (see, e.g., \cite{bar2023universal,chee2020efficient,gabrys2020locally,heckel2019characterization,hossein2016wealy,levy2019mutually,tabatabaei2015rewritable,wang2022coding,yazdi2015dna,yazdi2018mutually}).

One of the main concerns of this paper is to study the {\it finite-type constrained codes} (see below for a formal definition) for which the constraints can be characterized by a {\it finite} set of {\it forbidden substrings}, where by finite we mean that the number of forbidden substrings is a constant that does not increase with the code length. The study of these codes is an important direction in constrained coding (see \cite{marcus2001introduction} and the references therein). Recently, mostly motivated by applications in modern data storage systems, they have attracted a considerable amount of attention from both academia and industry. Next, let us briefly introduce several well-studied finite-type constrained codes in the literature.

\begin{itemize}
    \item The {\it $(d,k)$-run-length-limited} ($(d,k)$-RLL for short) constraint over $\Sigma_2=\{0,1\}$ requires that consecutive 1's are separated by at least $d$ and at most $k$ zeros. Since the repetition of identical nucleotides of length greater than six would result in a huge increase in substitution and deletion errors, the RLL constraint is introduced to reduce the error probability in DNA synthesizing and sequencing \cite{heckel2019characterization, ross2013characterizing}.

    \item The {\it $(\ell,\delta)$-locally-balanced} ($(\ell,\delta)$-LB for short) constraint requires that every substring of length $\ell$ has Hamming weight at least $\ell/2-\delta$ and at most $\ell/2+\delta$. This constraint is motivated by the observation that DNA sequences with 50$\%$ GC content are more stable than those with other GC contents (see \cite{gabrys2020locally,wang2022coding} for details).

    \item A string $u=u_1\cdots u_n$ is called a {\it palindrome} if $u=u^R:=u_n\cdots u_1$. The {\it $\ell$-palindrome-avoiding} ($\ell$-PA for short) constraint avoids the appearance of palindromes of length exactly $\ell$. This constraint is motivated by the observation that palindromes have a negative impact on the DNA sequencing pipelines (see \cite{bar2023universal} for details).

     \item For a string $u=u_1\cdots u_n\in\Sigma_q^*$, a positive integer $p$ is called a {\it period} of $u$ if for every $1\le i\le |u|-p$, $u_i=u_{i+p}$. Given $\ell$ and $p$, a string $u\in\Sigma_q^*$ is called {\it $(\ell,p)$-least-periodicity-avoiding} ($(\ell,p)$-LPA for short) if for every $1\le p'\le p-1$ and $1\le i\le |u|-\ell+1$, $p'$ is not a period of the substring $u_i\cdots u_{i+\ell-1}$. This is a new finite-type constraint with fundamental connection to codes designed for racetrack memories \cite{chee2018coding, sima2019correcting}.
\end{itemize}

It is not hard to observe that all of the constraints listed above can be realized by avoiding a finite set of substrings. For example, a binary code is $(1,3)$-RLL if and only if every codeword contains neither 11 nor 0000 as substrings. We will return to these constraints in \cref{section:applications} below.

To formalize, let $\Sigma_q$ be an alphabet of size $q$, $\Sigma_q^n$ be the set of all strings of length $n$ defined on $\Sigma_q$, and $\Sigma_q^*:=\cup_{i=0}^{\infty}\Sigma_q^i$, including the empty string of zero length. For $\tau,\omega\in\Sigma_q^*$, $\tau$ is said to be a {\it substring} of $\omega$ if there exist $u,v\in\Sigma_q^*$ (possibly empty) such that $\omega=u\tau v$. If $\tau$ is not a substring of $\omega$, then $\omega$ is called {\it $\tau$-free}. A set $C\subseteq\Sigma_q^*$ is $\tau$-free if every string in $C$ is $\tau$-free. For two sets $F,C\subseteq\Sigma_q^*$, $C$ is {\it $F$-free} if it is $\tau$-free for every $\tau\in F$. Using this notation, a set $C\subseteq\Sigma_2^*$ is $(1,3)$-RLL if and only if it is $\{11,0000\}$-free, $(4,1)$-LB if and only if $\{0000,1111\}$-free, $3$-PA if and only if $\{000,111,101,010\}$-free, and $(4,3)$-LPA if and only if $\{0000,1111,1010,0101\}$-free.

Given a subset $F\subseteq\Sigma_q^*$, let $C_F$ be the set of all $F$-free strings in $\Sigma_q^*$. If $F$ is finite, then we call $C_F$ a {\it finite-type constrained code} defined by the set $F$ of {\it constraints}. For an integer $n\ge 0$, let $C_F(n):=C_F\cap\Sigma_q^n$ and $N_F(n):=|C_F(n)|$. The main concern of this paper is to study the exact or asymptotic value of $N_F(n)$. For that purpose, let the {\it capacity} (or rate) of the constraints $F$ be
\begin{align}\label{eq:capacity}
 \ca(F)=\lim_{n\to\infty}\frac{\log_q N_F(n)}{n}.
\end{align}

\noindent Note that since the function $\log_q N_F(n)$ is subadditive; i.e., for integers $m,n\ge 1$, $\log_q N_F(n+m)\le \log_q N_F(n)+\log_q N_F(m)$, by Fekete's Lemma \cite{fekete1923verteilung} the limit $\lim_{n\to\infty}\frac{\log_q N_F(n)}{n}$ always exists.

Previously, it is known that $\ca(F)$ can be determined by the {\it spectral method} (see \cite{marcus2001introduction} and \cref{subsec:bridge} below for more details), which is based on the Perron-Frobenius theorem \cite{frobenius1912matrizen, perron1907zur}. We show that there is an alternative method arising from enumerative combinatorics, known as the {\it Goulden-Jackson cluster method} \cite{goulden1979inversion}, which can be used to determine not only the capacity, but also the exact formula for the cardinality of these codes.

To state our result, we will need a bit more definition. An infinite sequence $\{f_n\}_{n\ge 0}$, where $f_n\in\mathbb{R}$ for every $n\ge 0$, can be encoded as the coefficients of a formal power series $f(x)=\sum_{n\ge 0}f_nx^n$, called the {\it generating function} of the sequence. The {\it convergence radius} of $f$ is the largest non-negative real number $R$ such that the series converges if $|x|<R$ and diverges if $|x|>R$. If the series converges on the entire real axis, then $R=\infty$.

Our first main result is stated as follows. Note that we use $\ell(F):=\max_{\omega\in F}|\omega|$ to denote the maximum length of strings in $F$.

\begin{theorem}\label{thm:main}
    Let $F\subseteq\cup_{i\ge 2}\Sigma_q^i$ be a finite set of constraints. Then the following hold:
    \begin{enumerate}[(i)]
        \item Within time $O(|F|^3)$ one can compute two coprime polynomials $S(x)=\sum_{i=0}^sa_ix^i,T(x)=\sum_{i=0}^tb_ix^i\in\mathbb{Z}[x]$ with $\max\{s,t\}\le |F|\ell(F)$ such that the generating function $f_F(x):=\sum_{n\ge 0} N_F(n)x^{n}$ has a closed form expression $f_F(x)=\frac{T(x)}{S(x)}$;

        \item $\ca(F)=\log_q\frac{1}{R}=\log_q\frac{1}{x_0}$, where $R$ is the convergence radius of $f_F(x)$ and $x_0$ is the smallest positive root of $S(x)$;

        \item For any fixed constant $\epsilon\in(0,1)$, one can approximately compute $y^*\in(\ca(F)-\epsilon,\ca(F)+\epsilon)$ with time complexity $O(\log_q(2q\ln q/\epsilon) s^3)$.
    \end{enumerate}
\end{theorem}

The following proposition is an easy consequence of \cref{thm:main}.

\begin{proposition}\label{prop:exact-formula}
    Let $S(x),T(x)$ be defined as in \cref{thm:main} (i). Then the following hold:
    \begin{enumerate}[(i)]
        \item The sequence $\{N_F(n)\}_{n\ge 0}$ satisfies the recursive formula,
        \begin{align*}
            a_0N_F(n)+a_1N_F(n-1)+\cdots+a_sN_F(n-s)=b_n,
        \end{align*}
        where $b_n=0$ for all $n>t$;

        \item For every $n\ge 0$, $N_F(n)=\frac{1}{n!}\cdot\left(\frac{d^n}{dx}\frac{T}{S}\right)(0).$
    \end{enumerate}
\end{proposition}

To the best of our knowledge, before \cref{prop:exact-formula}, in the coding community, there was no systematic way to computer the recursive formula for $N_F(n)$. For example, in \cite{wang2022coding} the authors provided a lengthy proof to find only the recursive formula for the size of $(6,1)$-LB constrained codes. In \cref{section:applications} below we will apply \cref{thm:main} and \cref{prop:exact-formula} to compute the recursive formula and precise expression of $N_F(n)$ for several finite-type constraints mentioned in this subsection.

Moreover, the spectral method and the cluster method are closely related. We will discuss their connection in detail in \cref{subsec:bridge} below.

\subsection{Variable-length non-overlapping codes}

In the last subsection, we demonstrated that the cluster method can be used to study the maximum cardinality of finite-type constrained codes. To show the flexibility of the this method, we also use it to give an explicit upper bound on the maximum cardinality of {\it variable-length non-overlapping codes} \cite{bilotta2017variable, wang2022qary,
wang2024maximum}, which is a constrained code that is defined by an {\it infinite} number of forbidden substrings.

More precisely, for a string $\omega\in\Sigma_q^*$, let $|\omega|$ denote the length of $\omega$. For two positive integers $1\le i\le j\le |\omega|$, let $\omega(i,j)$ represent the substring of $\omega$ that begins at the $i$-th coordinate and ends at the $j$-th coordinate. We also write $\omega(i):=\omega(i,i)$ for short. Let $Pre(\omega):=\{\omega(1,k):1\le k<|\omega|\}$ be the set of all prefixes of $\omega$, and $Suf(\omega):=\{\omega(k,|\omega|):1< k\le |\omega|\}$ be the set of all suffixes of $\omega$, where we call $\omega(1,k)$ and $\omega(k,|\omega|)$ the {\it prefix} and the {\it suffix} of $\omega$, respectively. A set $C\subseteq \Sigma_q^n$ is called {\it non-overlapping}\footnote{It is also called {\it cross-bifix-free} in \cite{bajic2007construction,barcucci2014cross,bilotta2012new,chee2013cross} and {\it mutually correlated} in \cite{hossein2016wealy,levy2019mutually,yazdi2018mutually}. } if
for every two strings $u,v\in C$ (not necessarily distinct), $Pre(u)\cap Suf(v)=\emptyset$.

Non-overlapping codes were first introduced by Levenshtein \cite{levenshtein1964decoding} in 1964, motivated by applications in synchronization.
Yazidi {\it et al.} \cite{tabatabaei2015rewritable,yazdi2015dna} developed a random-access and rewritable DNA-based storage architecture, which used {\it constrained} non-overlapping codes, i.e., non-overlapping codes satisfying some other constraints, as address sequences. Inspired by their work, several variants of non-overlapping codes were studied in \cite{chee2020efficient,hossein2016wealy,levy2019mutually,yazdi2018mutually}.

We are interested in a variable-length variant of non-overlapping codes introduced by Bilotta \cite{bilotta2017variable}.
\begin{definition}\label{def:non-overlapping}
    A code $C\subseteq\Sigma_q^*$ is said to be a {\it variable-length non-overlapping code} if it is non-overlapping and no string in $C$ is a substring of another string in $C$.
\end{definition}

\noindent Bilotta noticed that variable-length non-overlapping codes have an interesting advantage when being used as address sequences: codewords of different lengths have nonzero Levenshtein distance for free.

Despite its potential practical applications, the study of variable-length codes is also interesting in its own right. Indeed, it is quite natural to study the variable-length analogs of combinatorial objects, and one of the most prominent examples is the celebrated LYM inequality (see \cite{bollobas1965generalized, lubell1966short, meshalkin1963generalization, yamamoto1954logarithmic}), which is a variable-length extension of the well-known Sperner's theorem \cite{sperner1928satz}.

Given two positive integers $q$ and $n$, let $C(q,n)$ (resp. $C(q,\le n)$) denote the maximum size of $q$-ary (resp. variable-length) non-overlapping codes of length $n$ (resp. at most $n$). The precise or asymptotic value of $C(q,n)$ has been extensively studied. Levenshtein \cite{levenshtein1964decoding} proposed a general upper bound on $C(q,n)$ by showing that
\begin{align}\label{ineq:upper bd}
    C(q,n)\le \left(\frac{n-1}{n}\right)^{n-1}\frac{q^n}{n}.
\end{align}

\noindent Blackburn \cite{blackburn2015non} showed that \eqref{ineq:upper bd} is tight when $n\mid q$. For other works on the upper and lower bounds of $C(q,n)$, see \cite{bajic2007construction,barcucci2014cross,blackburn2015non,bilotta2012new,chee2013cross,gilbert1960synchronization,levy2019mutually,qin2023non,wang2022qary}.

For $C(q,\le n)$, much less was known except for the obvious observation $C(q,n)\le C(q,\le n)$. Bilotta \cite{bilotta2017variable} proved a recursive inequality for $|C_n|:=|C\cap\Sigma_q^n|$ when $q=2$. Wang and Wang \cite{wang2022qary} presented a more refined analysis of that inequality and also generalized it to all $q$. Both Bilotta \cite{bilotta2017variable} and Wang and Wang \cite{wang2022qary} asked if there was a direct upper bound on $C(q,\le n)$.

Here, we answer their question affirmatively by showing that \eqref{ineq:upper bd} is also an upper bound for $C(q,\le n)$.

\begin{theorem}\label{thm:variable-length}
For integers $n\ge 2$ and $q\ge 2$,
\begin{align}\label{ineq:variable-length upper bd}
    C(q,\le n)\le\left(\frac{n-1}{n}\right)^{n-1}\frac{q^n}{n}.
\end{align}
\end{theorem}

Note that according to the aforementioned result of Blackburn \cite{blackburn2015non}, \cref{thm:variable-length} is tight when $n\mid q$. The proof of \cref{thm:variable-length} is a novel combination of the following three theorems, i.e., the Goulden-Jackson cluster theorem (see \cref{thm:GJ} below), the Cauchy-Hadamard theorem (see \cref{fact:real-analysis} (i) below), and the Vivanti-Pringsheim theorem (see \cref{fact:real-analysis} (ii) below). It indicates that the cluster method can be used to prove upper bounds for constrained coding problems that are not necessarily finite-type, showing the flexibility of the method.

\paragraph{Remark.}
The first author Y. Shen gave a talk on \cref{thm:variable-length} at the 10th National Conference on Combinatorics and Graph Theory, August 2022. When preparing this paper, we realized that Wang and Wang independently proved \cref{thm:variable-length}, using a quite different approach (see \cite{wang2024maximum,wang2024personal}).

\subsection{Bridging the cluster method and the spectral method}\label{subsec:bridge}

Since both the spectral method (see \cref{thm:spectral}) and the cluster method (see \cref{thm:main}) can determine $\ca(F)$ for a finite $F$. It is therefore a natural and intriguing question to study whether these two methods are inherently related.

Given a set of constraints $F\subseteq\Sigma_q^*$, the spectral method first connects the constrained code $C_F$ to a graph $G_F$, and then shows that the capacity $\ca(F)$ is equal to the logarithm of the spectral radius of $G_F$. More precisely, a {\it vertex-labeled directed graph} is a triple $G=(V,E,L)$, where $V$ is a set of vertices, $E$ is a set of ordered pairs of vertices in $V$, and $L:V\rightarrow\Sigma_q^\ell$ is a map that maps vertices in $V$ to $q$-ary strings of length $\ell$, where $\ell$ is some positive integer. A {\it walk} $\gamma$ in $G$ is a sequence of vertices $v_1\cdots v_n$ that satisfies for every $1\le i\le n-1$, $(v_i,v_{i+1})\in E$. The length of the walk $\gamma$ is the number of edges along the walk.

A vertex-labeled directed graph can be used to generate strings in $\Sigma_q^*$. For every walk $\gamma=v_1\cdots v_n$ in $G$, the string {\it generated} by $\gamma$ is defined by $L(\gamma):=\omega_1\omega_2(\ell)\cdots\omega_n(\ell)\in\Sigma_q^{\ell+n-1}$, where for every $1\le i\le n$, $\omega_i=L(v_i)\in\Sigma_q^\ell$, and $\omega_i(\ell)$ is the $\ell$-th coordinate of $\omega_i$.\footnote{Note that $\omega_1$ contributed all its $\ell$ coordinates to $L(\gamma)$, while each other $\omega_i$, where $2\le i\le n-1$, only contributed its last coordinate to $L(\gamma)$.} Let $$W_G:=\{L(\gamma):\gamma \text{ is a walk in $G$}\}\subseteq\Sigma_q^*$$ denote the set of all strings generated by $G$. Conversely, if a set $S\subseteq\Sigma_q^*$ is generated by some graph $G$, i.e., $S=W_G$, then we call $G$ a {\it graph representation} of $S$.\footnote{Note that a set $S$ could have more than one graph representations.}

The {\it adjacency matrix} $A$ of $G$ is a $|V|\times |V|$ matrix where for every $i,j\in V$, the $(i,j)$-th entry $A_{i,j}$ of $A$ is the number of edges directing from vertex $i$ to vertex $j$. The {\it spectral radius} of $G$, denoted by $\lambda(G)$, is the largest absolute value among the eigenvalues of $A$. Let $\det(A)$ be the {\it determinant} of the matrix $A$.

Let $C_F(\ge n):=\cup_{i\ge n}C_F(i)$ denote the set of all $F$-free strings of length at least $n$. A vertex-labeled directed graph is called {\it lossless} if every two distinct walks with the same initial and terminal vertices generate different strings. Next, we construct a lossless graph representation for $C_F(\ge \ell(F))$.

\begin{definition}\label{def:deBrujin-graph}
Let $G_F=(V,E,L)$ be a vertex-labeled directed graph defined as follows:
\begin{itemize}
  \item $V=C_F(\ell(F))$;
  \item $E=\{(\omega_1,\omega_2)\in V\times V:\omega_1(2,\ell(F))=\omega_2(1,\ell(F)-1)\}$;
  \item $L(v)=v$, for every vertex $v\in V$.
\end{itemize}
Note that the graph $G_F$ defined above is also known as the {\it de Brujin graph}.
\end{definition}

One can easily infer the following useful properties of $G_F$ via \cref{def:deBrujin-graph}.

\begin{fact}[see \cite{marcus2001introduction}]\label{fact:G_F}
    Let $F\subseteq\Sigma_q^*$ be a finite set of constraints and $G_F$ denote the vertex-labeled direct graph defined in \cref{def:deBrujin-graph}. Then the following hold:
    \begin{enumerate}[(i)]
        \item $G_F$ is lossless;
        \item 
        $W_{G_F}=C_F(\ge\ell(F))$;
        \item Every string $\omega\in W_{G_F}$ is generated by exactly one walk in $G_F$, and for every $n\ge\ell(F)$, $|W_{G_F}\cap\Sigma_q^n|$ equals to the number of $(n-\ell(F))$-length walks in $G_F$.
    \end{enumerate}
\end{fact}

The following theorem connects $\ca(F)$ to the spectral radius of $G_F$.

\begin{theorem}[see \cite{marcus2001introduction}]\label{thm:spectral}
    Let $F\subseteq\cup_{i\ge 2}\Sigma_q^i$ be a finite set of constraints and $G_F$ be the graph defined in \cref{def:deBrujin-graph}. Then
    $\ca(F)=\log_q\lambda(G_F)$.
\end{theorem}

The third main result of this paper shows that \cref{thm:main} and \cref{thm:spectral} are inherently related.

\begin{theorem}\label{thm:connection}
    Let $F\subseteq\cup_{i\ge 2}\Sigma_q^i$ be a finite set of constraints, $G_F$ be the graph representation of $C_F(\ge\ell(F))$ defined in \cref{def:deBrujin-graph}, and $A$ be the adjacency matrix of $G_F$. Let $J$ be the all one matrix. Then the following hold:
        \begin{enumerate}[(i)]
            \item The generating function $f_F(x)$ is
                \begin{align*}
                    f_F(x)=\sum\limits_{n=0}^\infty N_F(n)x^n=h_F(x)-\frac{x^{\ell(F)}\det(I-xA-J)}{\det(I-xA)},
                \end{align*}
                where $h_F(x):=\sum_{n=0}^{\ell(F)-1}N_F(n)x^n+x^{\ell(F)}$;

           \item $1/\lambda(G_F)$ is the convergence radius $R$ of $f_F(x)$.
        \end{enumerate}
    \end{theorem}

\paragraph{A few remarks on $F$.} To avoid triviality, we will put some restrictions on $F$. First, we say that $F$ is {\it non-degenerate}, if the set $\{n:N_F(n)\ge 1\}$ is unbounded, that is, there exists an infinite sequence $\{n_i\}_{i\ge 0}$ such that $N_F(n_i)\ge 1$ for each $i$. Second, it is easy to see that if there exist two strings $u,v\in F$ such that $v$ is a substring of $u$, then $C_F=C_{F\setminus\{u\}}$. Therefore, we assume that $F$ is {\it minimal} in the sense that no string in $F$ is a substring of another string in $F$; in particular, $F$ does not contain the empty string. Moreover, we assume that $F\cap\Sigma_q=\emptyset$, since otherwise we can instead consider strings defined on a smaller alphabet. We call a subset $F$ of strings {\it reduced} if $F\subseteq\cup_{i\ge2}\Sigma_q^i$ and no string in $F$ is a substring of another string in $F$. Throughout, we assume that the set $F$ of constraints is non-degenerate and reduced. Note that the set $F$ mentioned in  \cref{section:non-overlapping} (which devotes to the proof of \cref{thm:variable-length}) may be not finite.

\paragraph{Outline of the paper.} 
In \cref{section:constrained_coding}, we will introduce the main tool of this paper, i.e., the Goulden-Jackson cluster theorem, and present the proof of \cref{thm:main}. We will present the proofs of Theorems \ref{thm:variable-length} and \ref{thm:connection} in Sections \ref{section:non-overlapping} and \ref{section:connection}, respectively. As illustrations for the standard applications of the spectral method, in \cref{section:applications}, we will apply Theorems \ref{thm:main} and \ref{thm:GJ} to study the capacity and maximum cardinality for some finite-type constrained codes that are mentioned in \cref{subsec:finite-type}. We will conclude the paper in \cref{section:conclusion} and present a self-contained proof of \cref{thm:GJ} in \cref{section:appendix}.

\section{Constrained coding upper bounds via the cluster theorem}\label{section:constrained_coding}

In this section, we will introduce the Goulden-Jackson cluster theorem and use it to prove \cref{thm:main}. Let us begin with the definition of a cluster. Note that for two integers $m\le n$, let $[m,n]=\{m,\ldots,n\}$.

\begin{definition}[Cluster]\label{def:cluster}
Let $F\subseteq\cup_{i\ge 2}\Sigma_q^i$ be a finite reduced set of constraints. A string $\omega\in \Sigma_q^*$ together with a sequence of intervals $[i_1,j_1],\ldots,[i_{t},j_{t}]$, where $i_1, j_1,\ldots,i_{t},j_{t}$ are positive integers, is called a marked string, denoted by $(\omega; [i_1,j_1],\ldots,[i_{t},j_{t}])$, if the following two conditions hold:
\begin{enumerate}
[(i)]
  \item $i_1,\ldots,i_{t},j_1,\ldots,j_{t}\in [|\omega|], i_1<\cdots<i_{t}$, and $i_k<j_k$ for all $1\le k\le t$;
  \item $\omega(i_k,j_k)\in F$ for all $1\le k\le t$.
\end{enumerate}

\noindent In particular, each string $\omega\in\Sigma_q^*$ with empty interval also defines a marked word, denoted by $(\omega;\emptyset)$. Let $\mathcal{M}_F$ be the collection of all marked strings defined by $F$, that is, $\mathcal{M}_F$ is the collection of all $(\omega;[i_1,j_1],\ldots,[i_t,j_t])$'s that satisfy the above conditions (i) and (ii).

A marked string is said to be a cluster if it further satisfies the following:
\begin{enumerate}[(i)]
  \item [(iii)] $i_1=1, j_{t}=|\omega|$;
  \item [(iv)] $i_k < i_{k+1} \le j_k$ for all $1\le k \le t-1$.
\end{enumerate}

\noindent Let $\mathcal{C\ell}_F$ be the collection of all clusters defined by $F$, that is, $\mathcal{C\ell}_F$ is the collection of all $(\omega;[i_1,j_1],\ldots,[i_t,j_t])$'s that satisfy all the above conditions (i)-(iv).
\end{definition}

For example, let $\Sigma_1=\{\alpha\}$ and $F=\{\alpha\alpha\alpha\}$. Consider the following marked strings
\begin{center}
    $(\alpha\alpha\alpha\alpha\alpha\alpha;[1,3],[2,4],[4,6])$, $(\alpha\alpha\alpha\alpha\alpha\alpha;[1,3],[2,4],[3,5],[4,6])$, and $(\alpha\alpha\alpha\alpha\alpha\alpha;[1,3],[4,6])$,
\end{center} as depicted below:

\begin{equation}\nonumber
\tikzset{every picture/.style={line width=0.75pt}} 
\begin{tikzpicture}[x=0.75pt,y=0.75pt,yscale=-1,xscale=1]

\draw  [color={rgb, 255:red, 208; green, 2; blue, 27 }  ,draw opacity=1 ] (8,19) .. controls (8,13.48) and (18.3,9) .. (31,9) .. controls (43.7,9) and (54,13.48) .. (54,19) .. controls (54,24.52) and (43.7,29) .. (31,29) .. controls (18.3,29) and (8,24.52) .. (8,19) -- cycle ;
\draw  [color={rgb, 255:red, 74; green, 144; blue, 226 }  ,draw opacity=1 ] (22,19) .. controls (22,13.48) and (32.3,9) .. (45,9) .. controls (57.7,9) and (68,13.48) .. (68,19) .. controls (68,24.52) and (57.7,29) .. (45,29) .. controls (32.3,29) and (22,24.52) .. (22,19) -- cycle ;
\draw  [color={rgb, 255:red, 245; green, 166; blue, 35 }  ,draw opacity=1 ]
(54,19) .. controls (54,13.48) and (64.3,9) .. (77,9) .. controls (89.7,9) and (100,13.48) .. (100,19) .. controls (100,24.52) and (89.7,29) .. (77,29) .. controls (64.3,29) and (54,24.52) .. (54,19) -- cycle ;
\draw  [color={rgb, 255:red, 208; green, 2; blue, 27 }  ,draw opacity=1 ]
(115,18) .. controls (115,12.48) and (125.3,8) .. (138,8) .. controls (150.7,8) and (161,12.48) .. (161,18) .. controls (161,23.52) and (150.7,28) .. (138,28) .. controls (125.3,28) and (115,23.52) .. (115,18) -- cycle ;
\draw  [color={rgb, 255:red, 74; green, 144; blue, 226 }  ,draw opacity=1 ]
(129,18) .. controls (129,12.48) and (139.3,8) .. (152,8) .. controls (164.7,8) and (175,12.48) .. (175,18) .. controls (175,23.52) and (164.7,28) .. (152,28) .. controls (139.3,28) and (129,23.52) .. (129,18) -- cycle ;
\draw  [color={rgb, 255:red, 245; green, 166; blue, 35 }  ,draw opacity=1 ]
(161,18) .. controls (161,12.48) and (171.3,8) .. (184,8) .. controls (196.7,8) and (207,12.48) .. (207,18) .. controls (207,23.52) and (196.7,28) .. (184,28) .. controls (171.3,28) and (161,23.52) .. (161,18) -- cycle ;
\draw  [color={rgb, 255:red, 65; green, 117; blue, 5 }  ,draw opacity=1 ]
(144,19) .. controls (144,13.48) and (154.3,9) .. (167,9) .. controls (179.7,9) and (190,13.48) .. (190,19) .. controls (190,24.52) and (179.7,29) .. (167,29) .. controls (154.3,29) and (144,24.52) .. (144,19) -- cycle ;
\draw  [color={rgb, 255:red, 208; green, 2; blue, 27 }  ,draw opacity=1 ] (221,19) .. controls (221,13.48) and (231.3,9) .. (244,9) .. controls (256.7,9) and (267,13.48) .. (267,19) .. controls (267,24.52) and (256.7,29) .. (244,29) .. controls (231.3,29) and (221,24.52) .. (221,19) -- cycle ;
\draw  [color={rgb, 255:red, 245; green, 166; blue, 35 }  ,draw opacity=1 ] (267,19) .. controls (267,13.48) and (277.3,9) .. (290,9) .. controls (302.7,9) and (313,13.48) .. (313,19) .. controls (313,24.52) and (302.7,29) .. (290,29) .. controls (277.3,29) and (267,24.52) .. (267,19) -- cycle ;
\draw (9,14.5) node [anchor=north west][inner sep=0.75pt]    {$\alpha \ \alpha \ \alpha \ \alpha \ \alpha \ \alpha $};
\draw (116,14.5) node [anchor=north west][inner sep=0.75pt]    {$\alpha \ \alpha \ \alpha \ \alpha \ \alpha \ \alpha $};
\draw (222,14.5) node [anchor=north west][inner sep=0.75pt]    {$\alpha \ \alpha \ \alpha \ \alpha \ \alpha \ \alpha $};
\end{tikzpicture}
\end{equation}
By \cref{def:cluster}, the first two mark strings are clusters but the last one is not (it violates (iv)).


To prove \cref{thm:main} (i), we need the following cluster theorem, which shows that there is an efficient way to compute the closed form of the generating function for the sequence $\{N_F(n)\}_{n\ge 0}$.

\begin{theorem}[Goulden-Jackson cluster theorem \cite{goulden1979inversion,noonan1999goulden}]\label{thm:GJ}
Let $F\subseteq\cup_{i\ge 2}\Sigma_q^i$ be a finite reduced set of constraints. Then the generating function of $\{N_F(n)\}_{\ge0}$ is
\begin{align}\label{eq:clu-thm}
    f_F(x)=\sum_{n\ge 0} N_F(n)x^{n}=\frac{1}{1-qx-g_F(x)},
\end{align}
where
\begin{align}\label{eq:cluster-gen}
g_F(x)=\sum_{(\omega;[i_1,j_1],\ldots,[i_{t},j_{t}])\in  \mathcal{C\ell}_F} (-1)^{t}x^{|\omega|}.
\end{align}
Moreover, $g_F(x)$ has a closed form expression, which is a rational function that can be computed by solving a linear system consisting of $F$ linear equations with $|F|$ variables (over the rational function field $\mathbb{Q}(x)$).
\end{theorem}

The function $g_F(x)$ is known as the {\it cluster generating function} defined by $F$. The reader is referred to \cite{goulden1979inversion,noonan1999goulden} for other versions of the cluster theorem. For completeness, in \cref{section:appendix} we include the short proof of \cref{thm:GJ} in \cite{noonan1999goulden}.

To prove \cref{thm:main} (ii), we will need some classic results on the convergence radius of a power series.
The following facts are well known.

\begin{fact}\label{fact:real-analysis}
Let $f(x)=\sum_{n\ge 0}f_nx^n$ be a real power series, and $R$ be the convergence radius of $f(x)$.
    \begin{itemize}
        \item [(i)] By the Cauchy-Hadamard theorem (see Theorem 2.5 in \cite{Cauchy-Hardmard}), we have $R=(\limsup\limits_{n\to\infty}|f_n|^{1/n})^{-1}$.

        \item [(ii)] If we further have $f_n\ge 0$ for every $n\ge 0$ and $R$ is finite, then by the Vivanti-Pringsheim theorem (see Theorem \uppercase\expandafter{\romannumeral4}.6 in \cite{flajolet2009analytic}), $R$ is also a root of $1/f(x)$.
    \end{itemize}
\end{fact}

Lastly, we introduce the ingredients needed to prove \cref{thm:main} (iii). Note first that for a polynomial $h(x)=x^n+\sum_{i=0}^{n-1}a_ix^i\in\mathbb{R}[x]$, its {\it companion matrix} $CM(h)$ is defined as
\begin{align*}
CM(h):=\left(
    \begin{matrix}
        0 &0 &\cdots &0 & -a_0\\
        1 &0 &\cdots &0 & -a_1\\
        0 &1 &\cdots &0 & -a_2\\
        \vdots & \vdots &  & \vdots& \vdots\\
        0 &0 & \cdots &1 & -a_{n-1}
    \end{matrix}\right)
\end{align*}

\noindent It is easy to see that $h$ is indeed the {\it characteristic polynomial} of $CM(h)$ and therefore the roots of $h$ are precisely the eigenvalues of $CM(h)$. 
While there are several known algorithms in the literature that can be used to calculate the eigenvalues of a matrix, we will use the QR algorithm \cite{francis1961qr, francis1962qr}. The QR algorithm iteratively approximates the eigenvalues of a given matrix via QR decomposition (see \cite{allaire2008numerical,stoer1980introduction,watkins1982understanding} for more details). For our purpose, we will need the following fact on the convergence rate and complexity of QR algorithm.

\begin{fact}[see \cite{francis1961qr, watkins1982understanding}]
\label{fact:QR}
    Let $A$ be an $n\times n$ real matrix with at least one positive real eigenvalue. Let $\lambda$ be the smallest positive eigenvalue of $A$. For every positive integer $k$, let $\lambda_k$ denote the approximate value of $\lambda$ in the $k$-th iteration of the QR algorithm. Then there exists a constant $0<\beta<1$ such that $|\lambda_k-\lambda|<\beta^k$. Moreover, the complexity of each iteration in the QR algorithm is $O(n^3)$.
\end{fact}

Now, we are in a position to give the formal proof of \cref{thm:main}.

\begin{proof}[Proof of \cref{thm:main}]
    According to the second half of \cref{thm:GJ}, by applying the standard Gaussian elimination, $g_F(x)$ and $f_F(x)$ can be computed in $O(|F|^3)$ time. Since $g_F(x)$ is rational, $f_F(x)$ has a closed form expression $f_F(x)=\frac{T(x)}{S(x)}$ follows immediately. We postpone the proof on the upper bound of $\max\{s,t\}$ to \cref{section:appendix2}.

    We proceed to prove (ii). Note that the subadditivity of $\log_q N_F(n)$ (see the discussion below \eqref{eq:capacity}) implies the existence of the limit $\lim_{n\to\infty}\frac{\log_q N_F(n)}{n}$ and hence also the limit $\lim_{n\to \infty} N_F(n)^{1/n}$. Applying \cref{fact:real-analysis} (i) with $f=f_F$ and $f_n=N_F(n)$ yields that
    \begin{align}\label{eq:limt_of_N_F}
        R=(\limsup\limits_{n\to\infty}N_F(n)^{1/n})^{-1}=\left(\lim_{n\to\infty}N_F(n)^{1/n}\right)^{-1},
    \end{align}

    \noindent which implies that
    \begin{align*}
        \ca(F)=\lim_{n\to\infty}\frac{\log_q N_F(n)}{n}=\log_q\left(\lim_{n\to\infty} N_F(n)^{1/n}\right)=\log_q\frac{1}{R}.
    \end{align*}
    This proves the first equality in (ii). To prove the second equality, it suffices to show that $R=x_0$. First, we claim that $R\in[1/q,1]$. On one hand, as $|N_F(n)|\le q^n$ for each $n$, it is clear by \eqref{eq:limt_of_N_F} that $R\ge 1/q$. On the other hand, as by assumption, $F$ is non-degenerate, we have that $\limsup\limits_{n\to\infty}N_F(n)^{1/n}\ge 1$ and hence $R\le 1$. Consequently, it follows from \cref{fact:real-analysis} (ii) that $R$ is a root of $1/f_F(x)$. Furthermore, since $R$ is the convergence radius of $f_F(x)$, for every $|x|<R$, the power series $f_F(x)=\sum_{n\ge 0}N_F(n)x^n$ converges and hence $1/f_F(x)\neq 0$. So $R$ is the smallest positive root of $1/f_F(x)$, and hence $S(x)$. This completes the proof of (ii).

    To prove (iii), as illustrated in \cref{fact:QR}, we can use the QR algorithm to approximately compute the smallest positive root $x_0$ of $S(x)$. Furthermore, we claim that for every real $x^*$ satisfying $|x^*-x_0|< \epsilon/2q\ln q$, we have $|\log_q(1/x^*)-\log_q(1/x_0)|<\epsilon$. Indeed, by the mean value theorem, there exists some real $\xi$ between $x^*$ and $x_0$ such that
    \begin{align*}
        \left|\log_q\frac{1}{x^*}-\log_q\frac{1}{x_0}\right|=\left|-\frac{\ln q}{\xi}\right|\cdot|x^*-x_0|<\frac{\epsilon\ln q}{2\xi q\ln q}<\frac{\epsilon\ln q}{2\ln q-\epsilon}\le \epsilon,
    \end{align*}
    where the second inequality holds since $x_0=R\in[1/q,1]$ and $$\xi>x_0-|x^*-x_0|\ge 1/q-|x^*-x_0|>1/q- \epsilon/(2q\ln q).$$ Therefore, it follows from \cref{fact:QR} that after at most $\log_\beta( \epsilon/2q\ln q)$ iterations, we can obtain $y^*:=\log_q(1/x^*)$ that is $\epsilon$-close to $\ca(F)=\log_q(1/x_0)$. Lastly, it is not hard to verify by \cref{fact:QR} that the above process requires time $O(\log_q(2q\ln q/\epsilon) s^3)$.
\end{proof}

\begin{proof}[Proof of \cref{prop:exact-formula}]
    Note that by \cref{thm:main} we have
    \begin{align*}
        (a_0+\cdots+a_sx^s)\cdot(N_F(0)+N_F(1)\cdot x+\cdots)=b_0+\cdots+b_tx^t.
    \end{align*}
    Then (i) follows by comparing the coefficient of $x^n$ on the left and right sides of the above equation for each $n\ge s$. Similarly, (ii) follows by taking the $n$-th derivative on both sides of $f_F(x)=\frac{T(x)}{S(x)}$.
\end{proof}

\section{The maximum cardinality of variable-length non-overlapping codes}\label{section:non-overlapping}

In this section, we will prove \cref{thm:variable-length} using \cref{thm:GJ} and \cref{fact:real-analysis}. For that purpose, we first show that for a non-overlapping set $F$ of strings, every cluster in $\mathcal{C\ell}_{F}$ contains exactly one interval.

\begin{lemma}\label{lem:clusters_non-overlapping}
    Let $F\subseteq\cup_{i=2}^\infty\Sigma_q^i$ be a reduced set of constraints. If $F$ is non-overlapping, then for every $(\omega;[i_1,j_1],\ldots,[i_t,j_t])\in C\ell_{F}$, we have that $t=1$.
\end{lemma}

\begin{proof}
    Suppose for contradiction that $(\omega; [i_1,j_1],\ldots,[i_{t},j_{t}])\in\mathcal{C\ell}_{F}$ for some $t\ge 2$. Then, by \cref{def:cluster} (ii), we have $\omega(i_1,j_1),\omega(i_2,j_2)\in F$. Moreover, by \cref{def:cluster} (iv), we have $i_1<i_2\le j_1$. In fact, we have $i_1<i_2\le j_1<j_2$, since otherwise $i_1<i_2<j_2\le j_1$ and $\omega(i_2,j_2)$ is a substring of $\omega(i_1,j_1)$, a contradiction. Therefore, the two strings $\omega(i_1,j_1),~\omega(i_2,j_2)$ in $F$ are overlapping (more precisely, the $(j_1-i_2+1)$-length suffix of $\omega(i_1,j_1)$ equals to the $(j_1-i_2+1)$-length prefix of $\omega(i_2,j_2)$), which is again a contradiction.
\end{proof}

Applying \cref{thm:GJ} and \cref{lem:clusters_non-overlapping} to $F$ defined above, one can easily compute the generating function $f_F(x)$.

\begin{lemma}\label{lem:f_F for non-overlapping}
    Let $F\subseteq\cup_{i=2}^\infty\Sigma_q^i$ be a reduced set of constraints. If $F$ is non-overlapping, then the generating function $f_F(x)$ is
    \begin{align}\label{eq:f_F}
        f_F(x)=\sum_{i\ge 0}N_{F}(i)x^i=\frac{1}{1-qx+\sum_{i=2}^{\ell(F)}|F\cap\Sigma_q^i|x^i}.
    \end{align}
    Furthermore, the convergence radius $R$ of $f_F(x)$ is finite and positive.
\end{lemma}

\begin{proof}
    It follows from \cref{lem:clusters_non-overlapping} that $\mathcal{C\ell}_{F}=\{(\omega;[1,|\omega|]):\omega\in F\}$. Then, the cluster generating function defined by $F$ is
\begin{align*}
    g_{F}(x)=\sum_{(\omega;[1,|\omega|])\in C\ell_F} (-1)^1x^{|\omega|}=-\sum_{i=2}^{\ell(F)}|F\cap\Sigma_q^i|x^i.
\end{align*}
    Applying \cref{thm:GJ}, one can easily obtain the closed form expression of $f_F(x)$, which is exactly \eqref{eq:f_F}.

    Since $F\subseteq\cup_{i=2}^\infty\Sigma_q^i$ is non-overlapping, it does not contain any repetition string $\alpha\cdots\alpha$ of length at least two, where $\alpha\in\Sigma_q$. Therefore, for each $n\ge 2$, the set of repetition strings of length $n$ is $F$-free, which implies that $N_F(n)\ge q$. On the other hand, it is clear that $N_F(n)\le q^n$. Consequently, by \cref{fact:real-analysis} we can conclude that $R$ is finite and positive, as needed.
\end{proof}

Now we are ready to present the proof of \cref{thm:variable-length}.

\begin{proof}[Proof of \cref{thm:variable-length}]
    Let $C\subseteq\cup_{i=1}^n\Sigma_q^i$ be a variable-length non-overlapping code with largest cardinality, i.e., $|C|=C(q,\le n)$. For $1\le i\le n$, let $C_i:=C\cap\Sigma_q^i$ denote the set of strings in $C$ with length exactly $i$. Suppose $C_1=\{\alpha_1,\ldots,\alpha_s\}\subseteq\Sigma_q$ and denote $C':=C\setminus C_1=\cup_{i=2}^{n} C_i$ (possibly $C_1=\emptyset$ and then $s=0$).

    Note that if $C'=\emptyset$ then $|C|=C(q,1)= q$. From now on we assume that $C'\neq\emptyset$. By definition, every string in $C'$ contains none of the letters in $C_1$. Therefore, $C'$ is a non-overlapping code in $\cup_{i=2}^{n}\Sigma_{q-s}^i$, where we assume for simplicity that $\Sigma_q\setminus C_1=\Sigma_{q-s}=\{\beta_1,\ldots,\beta_{q-s}\}$.
    We will apply \cref{lem:f_F for non-overlapping} with $C'$ playing the role of $F$. Then, the generating function $f_{C'}(x)$ is
    \begin{align*}
         f_{C'}(x)=\sum_{i\ge 0}N_{C'}(i)x^i=\frac{1}{1-(q-s)x+\sum_{i=2}^{n}|C_i|x^i}:=\frac{1}{f(x)}.
    \end{align*}

    By \cref{lem:f_F for non-overlapping} and \cref{fact:real-analysis} (ii), it is not hard to see that the convergence radius $R$ of $f_{C'}(x)$ is a positive real root of $f(x)$. We further claim that $R\in(0,1)$. Indeed, as the set $\{\beta_1\beta_2,\ldots,\beta_1\beta_{q-s}\}$ is non-overlapping and $C$ is maximum, it is clear that $\sum_{i=2}^n|C_i|\ge q-s-1$. Therefore, $f(x)>0$ for every $x>1$, which implies that $R\in(0,1]$. If $R=1$, then $\sum_{i=2}^n|C_i|= q-s-1$ as $f(R)=0$. So, $|C|=|C_1|+\sum_{i=2}^n|C_i|=q-1<C(q,1)$, contradicting the maximality of $C$. Therefore, $R\in(0,1)$, as needed.

    As $f(R)=0$, one can easily see that $q-s=\sum_{i=2}^n|C_i|{R}^{i-1}+\frac{1}{R}$. Consequently,
    \begin{align*}
        q&=q-s+s=|C_n|R^{n-1}+|C_{n-1}|R^{n-2}+\dots+|C_2|R+|C_1|+\frac{1}{R}\\
        &\ge(|C_n|+|C_{n-1}|+\dots+|C_2|+|C_1|)R^{n-1} +\frac{1}{R}\nonumber\\
        &=\left(\sum_{i=1}^n |C_i|\right)R^{n-1}+\frac{1}{(n-1)R}+\cdots+\frac{1}{(n-1)R}\nonumber\\
        &\ge n\cdot \left(\sum_{i=1}^n|C_i|\cdot\frac{1}{(n-1)^{n-1}}\right)^{\frac{1}{n}}=n\cdot\left(|C|\cdot\frac{1}{(n-1)^{n-1}}\right)^{\frac{1}{n}}\nonumber,
    \end{align*}
where the first inequality follows from $R\in(0,1)$ and the second follows from the AM-GM inequality. Now, the theorem follows by rearranging the above inequality.
\end{proof}

According to the above proof, it is not hard to see that the inequality $|C|=\left(\frac{q}{n}\right)^n(n-1)^{n-1}$ holds only if $|C_1|=\cdots=|C_{n-1}|=0$ and $|C_n|=\frac{1}{(n-1)R^{n}}$.

\section{Bridging the cluster method and the spectral method}\label{section:connection}



The goal of this section is to present the proof of \cref{thm:connection}. Let $G$ be a graph defined on the vertex set $V=[m]$ and $A$ be the adjacency matrix of $G$. For an integer $n\ge 0$ and two vertices $i,j\in V$, let $A^n$ be the $n$-th power of $A$ and $A^n_{i,j}$ be the $(i,j)$-entry of $A^n$. We set $A^0=I$, which is the identity matrix.  It is well known and easy to see that $A_{ij}^n$ is the number of walks in $G$ of length $n$ with initial vertex $i$ and terminal vertex $j$; for notational convenience, we call these walks $(i,j)$-walks. For $i,j\in V$, we define the generating function for $(i,j)$-walks in $G$ as

\begin{align}\label{eq:gen-walks}
    F_{ij}(G;x)=\sum_{n\ge0}A_{ij}^nx^n.
\end{align}

The following fact gives a closed form expression for $F_{ij}(G;x)$, which will be useful in the proof of \cref{thm:connection} (i).

\begin{fact}[see Theorem 4.7.2 in \cite{stanley2011enumerative}]\label{fact:F_ij(G;x)}
Let $G$ be a vertex-labeled directed graph and $A$ be the adjacency matrix of $G$. Then for every two vertices $i,j$ of $G$, we have
\begin{align*}
  F_{ij}(G;x)=\frac{(-1)^{i+j}\det(I-xA;j,i)}{\det(I-xA)}.
\end{align*}
\end{fact}

Next, we will prove \cref{thm:connection} (i). 

\begin{proof}[Proof of \cref{thm:connection} (i)]
Recall the graph $G_F$ defined in \cref{def:deBrujin-graph}. By definition, $G_F$ has $m:=N_F(\ell(F))$ vertices. Without loss the generality, assume that $G_F$ is defined on the vertex set $V=[m]$.


We claim that the following equality holds for every integer $n\ge \ell(F)$,


\begin{align}\label{eq:N_F(n)}
    N_F(n)= \sum_{i,j\in[m]}A_{ij}^{n-\ell(F)}.
\end{align}
Indeed, for each $n\ge\ell(F)$, it follows directly from \cref{fact:G_F} (ii) that $N_F(n)=|W_{G_F}\cap\Sigma_q^n|$. Moreover, by \cref{fact:G_F} (iii), it is clear that $|W_{G_F}\cap\Sigma_q^n|=\sum_{i,j\in[m]}A_{ij}^{n-\ell(F)}$, which proves the claim.


By \eqref{eq:N_F(n)}, one can infer that
\begin{align}
f_F(x)&=\sum_{n\ge 0}N_F(n)x^n=\sum\limits_{n=0}^{\ell(F)-1}N_F(n)x^n+x^{\ell(F)}\sum\limits_{n=\ell(F)}^\infty \sum\limits_{i,j\in[m]}A_{ij}^{n-\ell(F)}x^{n-\ell(F)}\notag\\
&=h_F(x)-x^{\ell(F)}+x^{\ell(F)}\sum_{i,j\in[m]}F_{ij}(G;x)\notag\\
&=h_F(x)-x^{\ell(F)}+x^{\ell(F)}\sum_{i,j\in[m]}\frac{(-1)^{i+j}\det(I-xA;j,i)}{\det(I-xA)}\label{eq:f_F_det(I-xA)}
\end{align}
where the third equality follows from \eqref{eq:gen-walks} and the fourth equality follows from \cref{fact:F_ij(G;x)}.

For each $j\in[m]$, let $(I-xA)_j$ be the $j$-th column string of $I-xA$ and let $\mathbbm{1}$ be the all-one string of length $m$. We have
    \begin{align}
        &\det(I-xA-J)=\det\big((I-xA)_1-\mathbbm{1},\ldots,(I-xA)_m-\mathbbm{1}\big)\notag\\
        =&\det(I-xA)-\sum\limits_{j=1}^m\det\big((I-xA)_1,\ldots,(I-xA)_{j-1},\mathbbm{1},(I-xA)_{j+1},\ldots,(I-xA)_m\big)\notag\\
        =&\det(I-xA)-\sum\limits_{i,j\in[m]}(-1)^{i+j}\det(I-xA;i,j)\label{eq:det(I-xA-J)},
    \end{align}
where the second equality follows from the definition of determination and the third equality follows by expanding  $\det((I-xA)_1,\ldots,(I-xA)_{j-1},\mathbbm{1},(I-xA)_{j+1},\ldots,(I-xA)_m)$ according to its $j$-th column.

Clearly \cref{thm:connection} (i) follows from combining \eqref{eq:f_F_det(I-xA)} and \eqref{eq:det(I-xA-J)}.
\end{proof}

It remains to prove \cref{thm:connection} (ii). We will need the Perron-Frobenius theorem, as detailed below. Note that a real matrix (or string) is said to be {\it non-negative} if its entries are all non-negative.

\begin{theorem}[Perron-Frobenius theorem, see Theorem 3.15 in \cite{marcus2001introduction}]\label{thm:PF}
If $A\neq 0$ is a non-negative real square matrix, then the following hold: 
\begin{enumerate}[(i)]
    \item the spectral radius $\lambda(A)$ of $A$ is an eigenvalue of $A$;

    \item $\lambda(A)$ has a non-negative eigenstring $x\neq 0$. 
\end{enumerate}
\end{theorem}

Below, we will present the proof of \cref{thm:connection} (ii). 

\begin{proof}[Proof of \cref{thm:connection} (ii)]
Let $R$ be the convergence radius of $f_F(x)$. It is sufficient to show that $R=1/\lambda(G_F)$.

First, we will show $R\ge 1/\lambda(G_F)$. As $F$ is non-degenerate, it is clear that $A$ is non-negative and $A\neq0$. Indeed, if otherwise $A=0$, then it follows from \eqref{eq:N_F(n)} that for every $n\ge\ell(F)+1$, $N_F(n)=0$, which contradicts the assumption that $F$ is non-degenerate.
Moreover, it follows from \cref{thm:PF} (i) that the spectral radius $\lambda(G_F)$ is a root of the characteristic polynomial $\det(xI-A)$. Since $\det(I-xA)$ is the reversal polynomial of $\det(xI-A)$, i.e.,
\begin{align*}
    \det(I-xA)=x^n\cdot \det\left(\frac{1}{x}I-A\right),
\end{align*}
$1/\lambda(G_F)$ is a root of $\det(I-xA)$. In particular, $1/\lambda(G_F)$ is also the smallest positive root of $\det(I-xA)$ as $\lambda(G_F)$ is the largest positive root of $\det(xI-A)$. Recall that in the proof in \cref{thm:main} (ii) we have shown that $R$ is finite. Therefore, it follows from \cref{fact:real-analysis} (ii) that $R$ is a root of $1/f_F(x)$ and hence also a root of $\det(I-xA)$. We conclude that $R\ge 1/\lambda(G_F)$, as needed.

Second, we will show $R\le 1/\lambda(G_F)$.
It suffices to show that for any $\epsilon>0$, the series $\sum_{n\ge \ell(F)}N_F(n)\cdot\left(\frac{1}{\lambda(G_F)}+\epsilon\right)^n$ tends to infinity. By \cref{thm:PF} (ii), $\lambda(G_F)$ has some non-negative real eigenstring $x=(x_1,\ldots,x_m)^T$. Let $x_{\min}$ and $x_{\max}$ denote the minimal and the maximal positive coordinates of $x$, respectively. Note that possibly $x_{\min}=x_{\max}$.

Since for every positive integer $n$, $\lambda(G_F)^n$ is an eigenvalue of the matrix $A^n$ and $x$ is an eigenstring (i.e. $A^nx=\lambda(G_F)^nx$), for every $i\in [m]$, we have
\begin{align*}
\lambda(G_F)^nx_i= \left(\lambda(G_F)^nx\right)_i=\left(A^nx\right)_i
=\sum_{j=1}^mA_{ij}^nx_j\le x_{\max}\sum_{j=1}^{m}A_{ij}^n,
\end{align*}
where $\left(\lambda(G_F)^nx\right)_i$ and $\left(A^nx\right)_i$ denote the $i$-th coordinates of the strings $\lambda(G_F)^nx$ and $A^nx$, respectively.
Consequently,
\begin{equation}\label{ineq:CM-SM1}
\sum_{i,j\in[m]}A_{ij}^{n}\ge \frac{1}{x_{\max}}\sum_{i=1}^m\lambda(G_F)^{n}x_{i}\ge \frac{x_{\min}}{ x_{\max}}\cdot\lambda(G_F)^{n}.
\end{equation}

It follows from \eqref{eq:N_F(n)} and \eqref{ineq:CM-SM1} that for every $\epsilon>0$, we have
\begin{equation}\nonumber
\begin{aligned}
&\sum_{n\ge \ell(F)}N_F(n)\left(\frac{1}{\lambda(G_F)}+\epsilon\right)^n=\sum_{n=\ell(F)}^\infty\left(\sum_{i,j\in[m]}A_{ij}^{n-\ell(F)}\right)\left(\frac{1}{\lambda(G_F)}+\epsilon\right)^{n}\\
&\ge\sum_{n=\ell(F)}^\infty\frac{x_{\min}}{ x_{\max}}\lambda(G_F)^{n-\ell(F)}\left(\frac{1}{\lambda(G_F)}+\epsilon\right)^{n}\\
&=\frac{x_{\min}}{x_{\max}}\left(\frac{1}{\lambda(G_F)}+\epsilon\right)^{\ell(F)} \sum\limits_{n=0}^\infty\left(1+\epsilon\lambda(G_F)\right)^{n},
\end{aligned}
\end{equation}
which clearly tends to infinity. Therefore, we have $R\le 1/\lambda(G_F)$, as needed.
\end{proof}

\section{Upper bounds for some finite-type constrained codes}\label{section:applications}

\cref{thm:main} and \cref{thm:GJ} clearly work for every finite-type constrained code studied in the literature. For the purpose of illustration, in this section, we apply the two theorems above to study the capacity and cardinality of several finite-type constraint codes that have received considerable attention. For the ease of presentation, we restrict our attention to the binary alphabet $\Sigma_2=\{0,1\}$. Furthermore, we take $\epsilon$ in \cref{thm:main} (iii) to be $10^{-6}$, which implies that all capacities presented in this section are $10^{-6}$-close to the real value.

Note that by \cref{thm:GJ}, to compute the polynomials $S(x), T(x)$ in \cref{thm:main}, it suffices to solve a linear system of $|F|$ linear equations with $|F|$ variables, which can be accomplished by a Maple implementation. We refer the interested reader to \cite{noonan1999goulden} for more details.

\subsection{Palindrome-avoiding constraint}

Recall that for a positive integer $\ell$, a string $u=u_1\cdots u_{\ell}$ is called a palindrome if $u=u^R$, where $u^R=u_{\ell}\cdots u_1$. A code $C\subseteq\Sigma_2^*$ is said to be $\ell$-PA if every $u\in C$ does not contain any palindrome of length $\ell$. Let $$P_\ell:=\{u\in\Sigma_q^\ell:u=u^R\}.$$ Then a code $C$ is $\ell$-PA if and only if it is $P_{\ell}$-free.

The $\ell$-PA constraint was first introduced by Bar-Lev, Kobovich, Leitersdorf, and Yaakobi in \cite{bar2023universal}. They mainly considered the case where $\ell$ is dependent on $n$ and presented an encoding algorithm that can efficiently construct an $\ell$-PA code.

Consider, for example, the 6-PA constraint. According to the discussion above,
$$P_{6}=\{111111,000000,100001,010010,001100,110011,101101,011110\}.$$ It follows from \cref{thm:main} (i) (see also \cref{thm:GJ}) that the generating function $f_{P_{6}}(x)$ is
\begin{align*}
    f_{P_{6}}(x)=\frac{6x^{12}+8x^{11}+12x^{10}+16x^9+14x^8+14x^7+10x^6+10x^5+6x^4+3x^3+2x^2+x+1}{1-(2x^8+2x^6+2x^5+x^3+x)}.
\end{align*}
Applying \cref{thm:main} (iii) with $\epsilon=10^{-6}$ (see also \cref{fact:QR}) one can infer that
\begin{align*}
    |\ca(P_6)-0.7906315|\le 10^{-6}.
\end{align*}
By \cref{prop:exact-formula}, the recursive formula for $\{N_{P_{6}}(n)\}_{n\ge 0}$ is given by
\begin{align*}
\begin{matrix}
    &N_{P_6}(n)-N_{P_6}(n-1)-N_{P_6}(n-3)\\
    &-2N_{P_6}(n-5)-2N_{P_6}(n-6)-2N_{P_6}(n-8)
\end{matrix}
    =
    \begin{cases}
        16, & n=9;\\
        14, & n=7,8;\\
        12, & n=10;\\
        10, & n=5,6;\\
        8, & n=11;\\
        6, & n=4,12;\\
        3, & n=3;\\
        2, & n=2;\\
        1, & n=0,1;\\
        0, & \text{otherwise}.
    \end{cases}
\end{align*}
Lastly, again by \cref{prop:exact-formula}, we compute the exact values of $N_{P_{6}}(n)$ for $n=9,10,11,12,13,14,15,16$ in \cref{tab:PA}.

\begin{table}[h]
    \centering
    \begin{tabular}{|c|c|c|c|c|c|c|c|c|}
    \hline
        $n$ &$9$ &$10$ &  $11$ & $12$ & $13$ & $14$ & $15$&$16$\\
    \hline
        $N_{P_6}(n)$ &$294$& $508$& $878$ & $1518$& $2626$ & $4544$ & $7862$&$13600$\\
    \hline
    \end{tabular}
    \caption{The exact values of $N_{P_{6}}(n)$ for $n=9,10,11,12,13,14,15,16$.}
    \label{tab:PA}
\end{table}

\subsection{Least-periodicity-avoiding constraint}

A positive integer $p$ is said to be a period of a string $u_1\cdots u_n$ if for every $i\in[|u|-p]$, $u_i=u_{i+p}$. Recall that given two positive integers $\ell$ and $p$, a code $C\subseteq\Sigma_2^*$ is said to be $(\ell,p)$-LPA if for every $ p'\in[p-1]$ and every $u\in C$, $p'$ is not a period of any substring of $u$ with length $\ell$. Let $$A_{\ell,p}:=\{u\in\Sigma_2^\ell:\exists~p'\in[p-1], \text{ such that } u(i)=u(i+p') \text{ for every } i\in[\ell-p']\}.$$ By definition, a code $C\subseteq\Sigma_2^*$ is $(\ell,p)$-LPA if and only if it is $A_{\ell,p}$-free.

Kobovich, Leitersdorf, Bar-Lev and Yaakobi \cite{kobovich2022codes} proved a general upper bound on $\ca(A_{\ell,p})$ for all positive integers $\ell$, $p$ and $q$, which states that
\begin{align}\label{eq:LPA}
    \ca(A_{\ell,p})\le 1-\frac{(q-1)^2\log_qe}{2q^{\ell-p+3}},
\end{align}
where $e$ is the natural constant.

Consider the $(6,3)$-LPA constraint. According to the discussion above, $$A_{6,3}=\{111111,000000,101010,010101\}.$$ It follows from \cref{thm:main} (i) (see also \cref{thm:GJ}) that the generating function $f_{A_{6,3}}(x)$ is
\begin{align}
    f_{A_{6,3}}(x)=\frac{2x^5+x^4+x^3+x^2+x+1}{1-(x^4+x^3+x^2+x)}.
\end{align}
Applying \cref{thm:main} (iii) with $\epsilon=10^{-6}$ (see also \cref{fact:QR}) one can infer that
\begin{align*}
    |\ca(A_{6,3})-0.9467772|\le 10^{-6}.
\end{align*}
By \cref{prop:exact-formula}, the recursive formula for $\{N_{A_{6,3}}(n)\}_{n\ge 0}$ is given by
\begin{align*}
    N_{A_{6,3}}(n)-N_{A_{6,3}}(n-1)-N_{A_{6,3}}(n-2)-N_{A_{6,3}}(n-3)-N_{A_{6,3}}(n-4)=
    \begin{cases}
        2, & n=5;\\
        1, & n=0,1,2,3,4;\\
        0, & \text{otherwise}.
    \end{cases}
\end{align*}
Moreover, again by \cref{prop:exact-formula}, we compute the exact values of $N_{A_{6,3}}(n)$ for $n=9,10,11,12,13,14,15,16$ in \cref{tab:LPA}.

\begin{table}[h]
    \centering
    \begin{tabular}{|c|c|c|c|c|c|c|c|c|}
    \hline
        $n$ &$9$ &$10$ &  $11$ & $12$ & $13$ & $14$ & $15$&$16$\\
    \hline
        $N_{A_{6,3}}(n)$ &$432$& $832$& $1604$ & $3092$& $5960$ & $11488$ & $22144$&$42684$\\
    \hline
    \end{tabular}
    \caption{The exact values of $N_{A_{6,3}}(n)$ for $n=9,10,11,12,13,14,15,16$.}
    \label{tab:LPA}
\end{table}

Lastly, we estimate $\ca(A_{\ell,p})$ for some pairs of $(\ell,p)$, and compare our results with \eqref{eq:LPA} in \cref{tab:compair_LPA}.

\begin{table}[h]
    \centering
    \begin{tabular}{|c|c|c|c|c|c|c|}
    \hline
        $(\ell,p)$ & $(6,3)$&  $(6,4)$ &  $(7,2)$ & $(7,3)$& $(7,4)$\\
    \hline
        \cref{thm:main} (ii) & $0.94678$& $0.84397$ &  $0.98811$ &  $0.97523$&$0.93217$\\
    \hline
      Corollary 4 in \cite{kobovich2022codes}& $0.98873$& $0.97746$ & $0.99718$ &  $0.99436$& $0.98873$\\
      \hline
    \end{tabular}
    \caption{A comparison on the capacities of $(\ell,p)$-LPA constraint codes over $\Sigma_2$ between \cref{thm:main} (ii) and \eqref{eq:LPA}.}
    \label{tab:compair_LPA}
\end{table}

\subsection{Satisfying several constraints simultaneously}

As mentioned in \cref{subsec:finite-type}, the run-length-limited constraint, locally-balanced constraint and palindrome-avoiding constraint are all motivated by applications in DNA-based data storage. In this subsection, to show the flexibility and the power of our approach, we study the maximum cardinality of constrained codes that simultaneously satisfy all these three constraints.

Given two non-negative integers $d\le k$, a code $C\subseteq\Sigma_2^*$ is said to be $(d,k)$-RLL if every string $u\in C$ satisfies the condition that every consecutive 1's in $u$ are separated by at least $d$ and at most $k$ zeros. Let $$R_{d,k}:=\{ 10^t1: 0\le t\le d-1\}\cup\{0^{k+1}\}.$$ Note that a code $C\subseteq\Sigma_2^*$ is $(d,k)$-RLL if and only it is $R_{d,k}$-free.

For a string $u\in\Sigma_2^*$, let $wt(u)$ denote {\it the Hamming weight of $u$}, which is the number of non-zero coordinates of $u$. Given every two positive integers $\ell,\delta$, a string $u_1\cdots u_n\in\Sigma_2^*$ is said to be $(\ell,\delta)$-LB if for every $i\in[|u|-\ell+1]$, $\ell/2-\delta\le wt(u_i\ldots u_{i+\ell-1})\le\ell/2+\delta$, and a code $C\subseteq\Sigma_2^*$ is said to be $(\ell,\delta)$-LB if for every $u\in C$, every substring of $u$ with length $\ell$ is $(\ell,\delta)$-LB. Let $$L_{\ell,\delta}:=\left\{\omega\in\Sigma_2^{\ell}:wt(\omega)>\frac{\ell}{2}+\delta \text{ or } wt(\omega)<\frac{\ell}{2}-\delta\right\}.$$ It is not hard to see that a code $C\subseteq\Sigma_2^*$ is $(\ell,\delta)$-LB if and only if it is $L_{\ell,\delta}$-free.

Consider the combination of $(1,3)$-RLL constraint, $(6,1)$-LB constraint and $6$-PA constraint. Let
\begin{align*}
    D:=\{u\in L_{6,1}\cup P_6\cup R_{1,3}:\text{ there is no }v\in L_{6,1}\cup P_6\cup R_{1,3}\text{ such that $v\neq u$ and }v \text{ is a substring of }u\}.
\end{align*} Then a code $C\subseteq\Sigma_q^*$ is $(1,3)$-RLL, $(6,1)$-LB and $6$-PA if and only if it is $D$-free. It is easy to see that
\begin{align*}
   D= \{11,0000,010010,001000,000100\}.
\end{align*}
It follows from \cref{thm:main} (i) (see also \cref{thm:GJ}) that the generating function $f_{D}(x)$ is
\begin{align*}
f_{D}(x)=\frac{x^{13}+2 x^{12}+x^{11}-x^{10}-2 x^{9}-x^{8}-x^{7}-4 x^{6}-5 x^{5}-4 x^{4}-3 x^{3}-2 x^{2}-2 x -1}{x^{6}+x^2-1}.
\end{align*}
Applying \cref{thm:main} (iii) with $\epsilon=10^{-6}$ (see also \cref{fact:QR}) one can infer that
\begin{align*}
   | \ca(D)- 0.2757315|\le 10^{-6}.
\end{align*}
By \cref{prop:exact-formula}, the recursive formula on $\{N_{D}(n)\}_{n\ge 0}$ is given by
\begin{align*}
N_D(n)-N_D(n-2)-N_F(n-6)=
\begin{cases}
5,  & n=5;\\
4,  & n=4,6;\\
3,  & n=3;\\
2,  & n=1,2,9;\\
1, & n=0,7,8,10\\
-1, & n=11,13;\\
-2, & n= 12;\\
0, & \text{otherwise}.
\end{cases}
\end{align*}
Lastly, again by \cref{prop:exact-formula}, we compute the exact values of $N_{D}(n)$ for $n=9,10,11,12,13,14,15,16$ in \cref{tab:DNA}.

\begin{table}[h]
    \centering
    \begin{tabular}{|c|c|c|c|c|c|c|c|c|}
    \hline
        $n$ &$9$ &$10$ &  $11$ & $12$ & $13$ & $14$ & $15$&$16$\\
    \hline
        $N_{D}(n)$ &$20$& $24$& $29$ & $34$& $41$ & $50$ & $61$&$74$\\
    \hline
    \end{tabular}
    \caption{The exact values of $N_{D}(n)$ for $n=9,10,11,12,13,14,15,16$.}
    \label{tab:DNA}
\end{table}

\section{Concluding remarks}\label{section:conclusion}

We study the maximum cardinality of finite-type constrained codes. For readers interested in the encoding problem of these codes, we refer them to \cite{bar2023universal, chee2018coding, marcus2001introduction, sima2019correcting, wang2022coding}.

The codes studied in this paper are one-dimensional constrained system in the sense that they are strings over $\Sigma_q$. One can also consider two-dimensional constrained systems which are matrices over $\Sigma_q$ that forbid the appearance of some substructures. For instance, \cite{nguyen2021two,nguyen2021efficient,nguyen2022two, ordentlich2012low} considered two-dimensional {\it weight constrained matrices}, and \cite{barcucci2017non,barcucci20182d,barcucci2021non} considered two-dimensional {\it non-overlapping matrices}. In general, it is an intriguing question to develop a systematic method to study the maximum cardinality of finite-type high dimensional constrained systems.

{\small
\bibliographystyle{plain}
\bibliography{non-overlapping}

\begin{thebibliography}{10}

\bibitem{allaire2008numerical}
Gr{\'e}goire Allaire, Sidi~Mahmoud Kaber, Karim Trabelsi, and Gr{\'e}goire
  Allaire.
\newblock {\em Numerical linear algebra}, volume~55.
\newblock Springer, 2008.

\bibitem{bajic2007construction}
Dragana Bajic.
\newblock On construction of cross-bifix-free kernel sets.
\newblock {\em 2nd MCM COST}, 2100, 2007.

\bibitem{bar2023universal}
Daniella Bar-Lev, Adir Kobovich, Orian Leitersdorf, and Eitan Yaakobi.
\newblock Universal framework for parametric constrained coding.
\newblock {\em arXiv preprint arXiv:2304.01317}, 2023.

\bibitem{barcucci2017non}
Elena Barcucci, Antonio Bernini, Stefano Bilotta, and Renzo Pinzani.
\newblock Non-overlapping matrices.
\newblock {\em Theoretical Computer Science}, 658:36--45, 2017.

\bibitem{barcucci20182d}
Elena Barcucci, Antonio Bernini, Stefano Bilotta, and Renzo Pinzani.
\newblock A 2d non-overlapping code over aq-ary alphabet.
\newblock {\em Cryptography and Communications}, 10:667--683, 2018.

\bibitem{barcucci2021non}
Elena Barcucci, Antonio Bernini, Stefano Bilotta, and Renzo Pinzani.
\newblock Non-overlapping matrices via dyck words.
\newblock {\em ENUMERATIVE COMBINATORICS AND APPLICATIONS.}, 1:0--0, 2021.

\bibitem{barcucci2014cross}
Elena Barcucci, Stefano Bilotta, Elisa Pergola, Renzo Pinzani, and Jonathan
  Succi.
\newblock Cross-bifix-free sets via motzkin paths generation.
\newblock {\em arXiv preprint arXiv:1410.4710}, 2014.

\bibitem{bilotta2017variable}
Stefano Bilotta.
\newblock Variable-length non-overlapping codes.
\newblock {\em IEEE Transactions on Information Theory}, 63(10):6530--6537,
  2017.

\bibitem{bilotta2012new}
Stefano Bilotta, Elisa Pergola, and Renzo Pinzani.
\newblock A new approach to cross-bifix-free sets.
\newblock {\em IEEE Transactions on Information Theory}, 58(6):4058--4063,
  2012.

\bibitem{blackburn2015non}
Simon~R. Blackburn.
\newblock Non-overlapping codes.
\newblock {\em IEEE Transactions on Information Theory}, 61(9):4890--4894,
  2015.

\bibitem{bollobas1965generalized}
B{\'e}la Bollob{\'a}s.
\newblock On generalized graphs.
\newblock {\em Acta Mathematica Hungarica}, 16(3-4):447--452, 1965.

\bibitem{chee2013cross}
Yeow~Meng Chee, Han~Mao Kiah, Punarbasu Purkayastha, and Chengmin Wang.
\newblock Cross-bifix-free codes within a constant factor of optimality.
\newblock {\em IEEE Transactions on Information Theory}, 59(7):4668--4674,
  2013.

\bibitem{chee2018coding}
Yeow~Meng Chee, Han~Mao Kiah, Alexander Vardy, Van~Khu Vu, and Eitan Yaakobi.
\newblock Coding for racetrack memories.
\newblock {\em IEEE Transactions on Information Theory}, 64(11):7094--7112,
  2018.

\bibitem{chee2020efficient}
Yeow~Meng Chee, Han~Mao Kiah, and Hengjia Wei.
\newblock Efficient and explicit balanced primer codes.
\newblock {\em IEEE Transactions on Information Theory}, 66(9):5344--5357,
  2020.

\bibitem{fekete1923verteilung}
Michael Fekete.
\newblock {\"U}ber die verteilung der wurzeln bei gewissen algebraischen
  gleichungen mit ganzzahligen koeffizienten.
\newblock {\em Mathematische Zeitschrift}, 17(1):228--249, 1923.

\bibitem{flajolet2009analytic}
Philippe Flajolet and Robert Sedgewick.
\newblock {\em Analytic combinatorics}.
\newblock cambridge University press, 2009.

\bibitem{francis1961qr}
John G.~F. Francis.
\newblock The qr transformation a unitary analogue to the lr
  transformation-part 1.
\newblock {\em The Computer Journal}, 4(3):265--271, 1961.

\bibitem{francis1962qr}
John G.~F. Francis.
\newblock The qr transformation--part 2.
\newblock {\em The Computer Journal}, 4(4):332--345, 01 1962.

\bibitem{frobenius1912matrizen}
Von~G. Frobenius.
\newblock {\"U}ber matrizen aus nicht negativen elementen.
\newblock 1912.

\bibitem{gabrys2020locally}
Ryan Gabrys, Han~Mao Kiah, Alexander Vardy, Eitan Yaakobi, and Yiwei Zhang.
\newblock Locally balanced constraints.
\newblock In {\em 2020 IEEE International Symposium on Information Theory
  (ISIT)}, pages 664--669. IEEE, 2020.

\bibitem{gilbert1960synchronization}
Edgar~N. Gilbert.
\newblock Synchronization of binary messages.
\newblock {\em IRE Transactions on Information Theory}, 6(4):470--477, 1960.

\bibitem{goulden1979inversion}
Ian~P Goulden and David~M Jackson.
\newblock An inversion theorem for cluster decompositions of sequences with
  distinguished subsequences.
\newblock {\em Journal of the London Mathematical Society}, 2(3):567--576,
  1979.

\bibitem{heckel2019characterization}
Reinhard Heckel, Gediminas Mikutis, and Robert~N Grass.
\newblock A characterization of the dna data storage channel.
\newblock {\em Scientific reports}, 9(1):9663, 2019.

\bibitem{hossein2016wealy}
Syed~Mahamud Hossein, Tabatabaei Yazdi, Han~Mao Kiah, and Olgica Milenkovic.
\newblock Weakly mutually uncorrelated codes.
\newblock In {\em 2016 IEEE International Symposium on Information Theory
  (ISIT)}, pages 2649--2653, 2016.

\bibitem{immink1990runlength}
Kees A.~Schouhamer Immink.
\newblock Runlength-limited sequences.
\newblock {\em Proceedings of the IEEE}, 78(11):1745--1759, 1990.

\bibitem{immink1995efmplus}
Kees A.~Schouhamer Immink.
\newblock Efmplus: The coding format of the multimedia compact disc.
\newblock {\em IEEE Transactions on Consumer Electronics}, 41(3):491--497,
  1995.

\bibitem{immink2001survey}
Kees A.~Schouhamer Immink.
\newblock A survey of codes for optical disk recording.
\newblock {\em IEEE Journal on selected areas in communications},
  19(4):756--764, 2001.

\bibitem{immink2004codes}
Kees A.~Schouhamer Immink.
\newblock {\em Codes for mass data storage systems}.
\newblock Shannon Foundation Publisher, 2004.

\bibitem{immink1998codes}
Kees A.~Schouhamer Immink, Paul~H. Siegel, and Jack~K. Wolf.
\newblock Codes for digital recorders.
\newblock {\em IEEE Transactions on Information Theory}, 44(6):2260--2299,
  1998.

\bibitem{kobovich2022codes}
Adir Kobovich, Orian Leitersdorf, Daniella Bar-Lev, and Eitan Yaakobi.
\newblock Codes for constrained periodicity.
\newblock {\em arXiv preprint arXiv:2205.03911}, 2022.

\bibitem{Cauchy-Hardmard}
Serge Lang.
\newblock {\em Complex analysis}, volume 103 of {\em Graduate Texts in
  Mathematics}.
\newblock Springer-Verlag, New York, fourth edition, 1999.

\bibitem{levenshtein1964decoding}
Vladimir~Iosifovich Levenshtein.
\newblock Decoding automata which are invariant with respect to their initial
  state.
\newblock {\em Probl. Cybern}, 12:125--136, 1964.

\bibitem{levy2019mutually}
Maya Levy and Eitan Yaakobi.
\newblock Mutually uncorrelated codes for dna storage.
\newblock {\em IEEE Transactions on Information Theory}, 65(6):3671--3691,
  2019.

\bibitem{lubell1966short}
David Lubell.
\newblock A short proof of sperner's lemma.
\newblock {\em Journal of Combinatorial Theory}, 1(2):299, 1966.

\bibitem{marcus2001introduction}
Brian~H Marcus, Ron~M Roth, and Paul~H Siegel.
\newblock An introduction to coding for constrained systems.
\newblock {\em Lecture notes}, 2001.

\bibitem{meshalkin1963generalization}
Lev~D Meshalkin.
\newblock Generalization of sperner’s theorem on the number of subsets of a
  finite set.
\newblock {\em Theory of Probability \& Its Applications}, 8(2):203--204, 1963.

\bibitem{nguyen2021two}
Chi~Dinh Nguyen, Kui Cai, et~al.
\newblock Two-dimensional weight-constrained codes for crossbar resistive
  memory arrays.
\newblock {\em IEEE Communications Letters}, 25(5):1435--1438, 2021.

\bibitem{nguyen2021efficient}
Tuan~Thanh Nguyen, Kui Cai, Kees A~Schouhamer Immink, and Yeow~Meng Chee.
\newblock Efficient design of capacity-approaching two-dimensional
  weight-constrained codes.
\newblock In {\em 2021 IEEE International Symposium on Information Theory
  (ISIT)}, pages 2930--2935. IEEE, 2021.

\bibitem{nguyen2022two}
Tuan~Thanh Nguyen, Kui Cai, Han~Mao Kiah, Kees A~Schouhamer Immink, and
  Yeow~Meng Chee.
\newblock Two dimensional rc/subarray constrained codes: Bounded weight and
  almost balanced weight.
\newblock {\em arXiv preprint arXiv:2208.09138}, 2022.

\bibitem{noonan1999goulden}
John Noonan and Doron Zeilberger.
\newblock The goulden—jackson cluster method: extensions, applications and
  implementations.
\newblock {\em Journal of Difference Equations and Applications},
  5(4-5):355--377, 1999.

\bibitem{ordentlich2012low}
Erik Ordentlich and Ron~M Roth.
\newblock Low complexity two-dimensional weight-constrained codes.
\newblock {\em IEEE Transactions on Information Theory}, 58(6):3892--3899,
  2012.

\bibitem{perron1907zur}
Oskar Perron.
\newblock Zur theorie der matrices.
\newblock {\em Mathematische Annalen}, 64:248--263, 1907.

\bibitem{qin2023non}
Chunyan Qin, Bocong Chen, and Gaojun Luo.
\newblock On non-expandable cross-bifix-free codes.
\newblock {\em arXiv preprint arXiv:2309.08915}, 2023.

\bibitem{ross2013characterizing}
Michael~G Ross, Carsten Russ, Maura Costello, Andrew Hollinger, Niall~J Lennon,
  Ryan Hegarty, Chad Nusbaum, and David~B Jaffe.
\newblock Characterizing and measuring bias in sequence data.
\newblock {\em Genome biology}, 14:1--20, 2013.

\bibitem{sima2019correcting}
Jin Sima and Jehoshua Bruck.
\newblock Correcting deletions in multiple-heads racetrack memories.
\newblock In {\em 2019 IEEE International Symposium on Information Theory
  (ISIT)}, pages 1367--1371, 2019.

\bibitem{sperner1928satz}
Emanuel Sperner.
\newblock Ein satz {\"u}ber untermengen einer endlichen menge.
\newblock {\em Mathematische Zeitschrift}, 27(1):544--548, 1928.

\bibitem{stanley2011enumerative}
Richard~P Stanley.
\newblock Enumerative combinatorics volume 1 second edition.
\newblock {\em Cambridge studies in advanced mathematics}, 2011.

\bibitem{stoer1980introduction}
Josef Stoer, Roland Bulirsch, R~Bartels, Walter Gautschi, and Christoph
  Witzgall.
\newblock {\em Introduction to numerical analysis}, volume~2.
\newblock Springer, 1980.

\bibitem{yazdi2018mutually}
Sadegh M.~Hossein Tabatabaei~Yazdi, Han~Mao Kiah, Ryan Gabrys, and Olgica
  Milenkovic.
\newblock Mutually uncorrelated primers for dna-based data storage.
\newblock {\em IEEE Transactions on Information Theory}, 64(9):6283--6296,
  2018.

\bibitem{yazdi2015dna}
Sadegh M.~Hossein Tabatabaei~Yazdi, Han~Mao Kiah, Eva Garcia-Ruiz, Jian Ma,
  Huimin Zhao, and Olgica Milenkovic.
\newblock Dna-based storage: Trends and methods.
\newblock {\em IEEE Transactions on Molecular, Biological and Multi-Scale
  Communications}, 1(3):230--248, 2015.

\bibitem{tabatabaei2015rewritable}
Sadegh M.~Hossein Tabatabaei~Yazdi, Yongbo Yuan, Jian Ma, Huimin Zhao, and
  Olgica Milenkovic.
\newblock A rewritable, random-access dna-based storage system.
\newblock {\em Scientific reports}, 5(1):1--10, 2015.

\bibitem{wang2022coding}
Chen Wang, Ziyang Lu, Zhaojun Lan, Gennian Ge, and Yiwei Zhang.
\newblock Coding schemes for locally balanced constraints.
\newblock In {\em 2022 IEEE International Symposium on Information Theory
  (ISIT)}, pages 1342--1347, 2022.

\bibitem{wang2022qary}
Geyang Wang and Qi~Wang.
\newblock Q-ary non-overlapping codes: A generating function approach.
\newblock {\em IEEE Transactions on Information Theory}, 68(8):5154--5164,
  2022.

\bibitem{wang2024maximum}
Geyang Wang and Qi~Wang.
\newblock On the maximum size of variable-length non-overlapping codes.
\newblock {\em arXiv preprint arXiv:2402.18896}, 2024.

\bibitem{wang2024personal}
Qi~Wang.
\newblock Personal communication.
\newblock March, 2024.

\bibitem{watkins1982understanding}
David~S Watkins.
\newblock Understanding the qr algorithm.
\newblock {\em SIAM review}, 24(4):427--440, 1982.

\bibitem{yamamoto1954logarithmic}
Koichi Yamamoto.
\newblock Logarithmic order of free distributive lattice.
\newblock {\em Journal of the Mathematical Society of Japan}, 6(3-4):343--353,
  1954.

\end{thebibliography}
}

\appendix

\section{Proof of \texorpdfstring{\cref{thm:GJ}}{}}\label{section:appendix}

In this subsection, we will present the proof of \cref{thm:GJ} by Noonan and Zeilberger \cite{noonan1999goulden}. In particular, we will show how to obtain the closed form of $f_F(x)$ (see \eqref{eq:clu-thm}) by solving a linear system of $|F|$ equations with $|F|$ variables over the rational function field $\mathbb{Q}(x)$.

For every string $\omega\in\Sigma_q^*$, let $F(\omega)$ denote the set of substrings of $\omega$ in $F$. It follows from the definition of $f_F(x)$ that
\begin{align}
    f_F(x)&=\sum\limits_{n\ge0} N_F(n)x^n=\sum\limits_{\omega\in C_F}x^{|\omega|}=\sum\limits_{\omega\in\Sigma_q^*}0^{|F(\omega)|}x^{|\omega|}\notag\\
    &=\sum\limits_{\omega\in\Sigma_q^*}(1+(-1))^{|F(\omega)|}x^{|\omega|}\notag\\
    &=\sum\limits_{\omega\in\Sigma_q^*}\sum\limits_{S\subseteq F(\omega)}(-1)^{|S|}x^{|\omega|}\notag\\
    &=\sum\limits_{(\omega;[i_1,j_1],\ldots,[i_t,j_t])\in\mathcal{M}_F}(-1)^tx^{|\omega|}\label{eq:marked},
\end{align}
where the third equality follows from the fact that $0^t=0$ if $t>0$ and $0^t=1$ if $t=0$, the fifth equality follows from the binomial theorem, and the last equality follows from \cref{def:cluster}.\

For every two subsets $A,B\subseteq\Sigma_q^*$, let $AB:=\{uv\in\Sigma_q^*:u\in A, v\in B\}$. Note that every non-empty marked string ends in either a symbol in $\Sigma_q$ that is not part of a cluster, or a cluster. Hence, we can partition $\mathcal{M}_F$ as follows
\begin{align*}
    \mathcal{M}_F=\{\text{empty string}\}\cup \mathcal{M}_F\Sigma_q\cup\mathcal{M}_F\mathcal{C\ell}_F
\end{align*}
Consequently, it follows from \eqref{eq:cluster-gen} and \eqref{eq:marked} that
\begin{align*}
    f_F(x)=\sum\limits_{(\omega;[i_1,j_1],\ldots,[i_t,j_t])\in\mathcal{M}_F}(-1)^tx^{|\omega|}=1+f_F(x)\cdot qx+f_F(x)\cdot g_F(x),
\end{align*}
which implies \eqref{eq:clu-thm}.

Now, it remains to deduce the closed form of $g_F(x)$. Recall that for each string $\omega\in\Sigma_q^*$, $Pre(\omega)$ and $Suf(\omega)$ denote the sets of all prefixes and suffixes of $\omega$, respectively. For two strings $u,v\in\Sigma_q^*$, let $Overlap(u,v):=Suf(u)\cap Pre(v)$. For each $\omega\in Overlap(u,v)$, let $v\backslash\omega$ denote the substring of $v$ obtained by deleting $\omega$ from $v$, and let
\begin{align}\label{eq:h_{u,v}}
    h_{u,v}(x):=\sum_{ \omega\in Overlap(u,v)}x^{|v\backslash\omega|}.
\end{align}
Then for all pairs $u,v\in F$, $h_{u,v}$ can be explicitly computed and $1\le\deg(h_{u,v})\le\ell(F)-1$ (in particular, it does not have constant term).

One can partition $\mathcal{C}\ell_F$ into
\begin{align*}
    \mathcal{C}\ell_F=\bigcup\limits_{v\in F}\mathcal{C}\ell_F[v],
\end{align*}
where $\mathcal{C}\ell_F[v]$ denotes the set of all clusters whose last substring is $v$, i.e.,
$$\mathcal{C}\ell_F[v]:=\{(\omega;[i_1,j_1],\cdots,[i_t,j_t])\in\mathcal{C}\ell_F:\omega(i_t,j_t)=v\}.$$
This implies that
\begin{align}\label{eq:g_F}
    g_F(x)=\sum\limits_{v\in F}\sum\limits_{(\omega;[i_1,j_1],\cdots,[i_t,j_t])\in\mathcal{C}\ell_F[v]}(-1)^tx^{|\omega|}=\sum\limits_{v\in F}g_F^{(v)}(x),
\end{align}
where $g_F^{(v)}(x):=\sum_{(\omega;[i_1,j_1],\cdots,[i_t,j_t])\in\mathcal{C}\ell_F[v]}(-1)^tx^{|\omega|}$.

Note that for each cluster $(\omega;[i_1,j_1],\cdots,[i_t,j_t])\in\mathcal{C}\ell_F[v]$, either $t=1$ and $(\omega;[i_1,j_1])=\{(v;[1,|v|])\}$ or $t\ge 2$ and $\omega(i_t,j_t)=v$; in the latter case, removing $v$ from the original cluster yields a shorter cluster $(\omega(i_1,j_{t-1});[i_1,j_1],\cdots,[i_{t-1},j_{t-1}])$ that ends in some $u=\omega(i_{t-1},j_{t-1})\in F$ such that $Overlap(u,v)\neq\emptyset$.

Therefore, there exists a bijection between $\mathcal{C}\ell_F[v]$ and
$$\{(v;[1,|v|])\}\cup\bigcup\limits_{u\in F}\mathcal{C}\ell_F[u]\times Overlap(u,v).$$
Consequently, for each $v\in F$, we have that
\begin{align*}
    g_F^{(v)}(x)=-x^{|v|}-\sum\limits_{u\in F}g_F^{(u)}(x)h_{u,v}(x);
\end{align*}
or equivalently,
\begin{align}\label{eq:g_F^(v)}
    (1+h_{v,v}(x))\cdot g_F^{(v)}(x)+\sum\limits_{u\in F\setminus\{v\}}h_{u,v}(x)\cdot g_F^{(u)}(x)=-x^{|v|}.
\end{align}

If we view $\{g_F^{(v)}(x):v\in F\}$ as a set of variables, then \eqref{eq:g_F^(v)} can be viewed as a linear system of $|F|$ linear equations with $|F|$ variables, defined over $\mathbb{Q}(x)$. Suppose that $F=\{v_1,\ldots,v_{|F|}\}$. Let $B$ denote the coefficient matrix of \eqref{eq:g_F^(v)}, i.e., for every $v_i,v_j\in F$, the $(v_i,v_j)$-entry of $B$ is
\begin{align*}
    B_{v_i,v_j}=
        \begin{cases}
            1+h_{v_i,v_j}(x), & i=j;\\
            h_{v_i,v_j}(x), & i\neq j.
        \end{cases}
\end{align*}
As for all $v_i,v_j\in F$, $h_{v_i,v_j}$ does not have constant term, one can check by Leibniz formula that $\det(B)$ must contain the constant term 1 and it is therefore nonzero over $\mathbb{Q}(x)$. We conclude that by Gauss elimination, \eqref{eq:g_F^(v)} is solvable over $\mathbb{Q}(x)$ in time $O(|F|^3)$. This completes the proof of \cref{thm:GJ}.

\section{Proof of \texorpdfstring{\cref{thm:main} (i)}{}}\label{section:appendix2}

Lastly, we prove the upper bounds of $s,t$ in \cref{thm:main} (i). By the discussion in \cref{section:appendix}, it is clear that $\det(B)\in\mathbb{Z}[x]$ and $\deg(\det(B))\le |F|(\ell(F)-1)$. For each $j\in[|F|]$, let $B_j$ denote the matrix formed by replacing the $j$-th column string of $B$ with the column string $(-x^{|v_1|},\ldots,-x^{|v_{|F|}|})^T$. It is also clear that for each $j$, $\det(B_j)\in\mathbb{Z}[x]$ and $\deg(\det(B_j))\le \ell(F)+(|F|-1)(\ell(F)-1)\le|F|\ell(F)$. Therefore, it follows from the Cramer's rule and \eqref{eq:clu-thm} that the degrees of $S(x)$ and $T(x)$ are both bounded by $\ell(F)|F|$, as needed.

\end{document}